%% file: main.tex
\documentclass[11pt]{article}
\input{packages}

\usepackage{tcolorbox}

\input{defines}

\date{}

\begin{document}
\input{titleauth}

\maketitle
\sloppy

\begin{abstract}
    In the cup game, an adversary distributes $1$ unit of water among $n$ initially empty cups during each time step. The player then selects a single cup from which to remove up to $1$ unit of water, with the goal of minimizing the backlog, i.e., the supremum of the height of the fullest cup over all time steps.
    In the cup flushing game, the player is additionally allowed to empty the chosen cup entirely. Past work has shown that the optimal backlog in both of these settings is $\Theta(\log n)$. Furthermore, the \textbf{greedy} algorithm, which always removes water from the fullest cup, has been shown in previous work to be exactly optimal in both the cup game and the cup flushing game. 
    
    We introduce a new model, the semi-oblivious cup game, in which the player is uncertain of the exact height of each cup. We analyze the performance of the \textbf{greedy} algorithm in this setting, which can be viewed as selecting an arbitrary cup within a constant multiplicative factor of the fullest cup. We prove matching upper and lower bounds showing that the \textbf{greedy} algorithm achieves a backlog of $\Theta(n^{\frac{c-1}{c}})$ in the semi-oblivious cup game. We also establish matching upper and lower bounds of $2^{\Theta(\sqrt{\log n})}$ in the semi-oblivious cup flushing game. Finally, we show that in an additive error setting, greedy is actually able to achieve backlog $\Theta(\log n)$, via matching upper and lower bounds. All of our lower bounds apply for adaptive adversaries against any (even randomized) algorithm, proving that greedy is asymptotically optimal in the semi-oblivious model.

\end{abstract}

\newtheorem{claim}{Claim}

\thispagestyle{empty}

\newpage
\setcounter{page}{1}

\section{Introduction}
The \defn{cup game} is a classical buffer-management scheduling problem, dating back to the late 1960s~\cite{Liu69, LiuLa73}, which is embodied as a two-player game played on $n$ initially empty cups. During each time step, the 
\defn{adversary} first distributes $1$ unit of water among the cups arbitrarily.
The \defn{player} then selects a single cup, and is allowed to remove up to $1$ unit of water from that cup. As we will see, in some variations, the player is allowed to remove more than $1$ unit of water, but always just from a single cup.
The player's goal is to minimize the \defn{backlog}, the supremum of the amount of water in the fullest cup over all time steps.

Cup games have found ample applications in areas including processor scheduling ~\cite{BaruahCoPl96,GasieniecKlLe17,BaruahGe95,LitmanMo11,LitmanMo05,MoirRa99,BarNi02,GuanYi12,Liu69, LiuLa73, AdlerBeFr03, LitmanMo09, DietzRa91,BenderFaKu19, Kuszmaul20, tailsize, BenderKu20}, network-switch buffer management~\cite{Goldwasser10,AzarLi06,RosenblumGoTa04,Gail93}, quality-of-service guarantees~\cite{BaruahCoPl96,AdlerBeFr03,LitmanMo09}, and data-structure deamortization~\cite{AmirFaId95,DietzRa91,DietzSl87,AmirFr14,Mortensen03,GoodrichPa13,FischerGa15,Kopelowitz12,BenderDaFa20}. See \cite{Kuszmaul20} for a detailed discussion of the related work.
Since cup games have arisen independently several times due to their natural modeling of many real-world phenomena in computer systems, they have also been studied under the monikers the \defn{leaky-bucket problem}~\cite{AdlerBeFr03}, \defn{Cinderella versus the wicked stepmother}~\cite{BodlaenderHuKu12} and \defn{bamboo garden trimming}~\cite{GasieniecKlLe17} (as well as more restrained names such as \defn{multiprogram scheduling} \cite{LiuLa73}).

Most variations of the cup game have strengthened the player's hand \cite{AdlerBeFr03, BenderFaKu19, LitmanMo09, BaruahCoPl96,GasieniecKlLe17,BaruahGe95,LitmanMo11,LitmanMo05,MoirRa99,BarNi02,GuanYi12,Liu69,LiuLa73, BenderKu20, DietzSl87, Kuszmaul22, GasieniecJuKl24}.
A first modification is to give the player \defn{resource augmentation}, where the player is allowed to empty more than $1$ unit of water. In the limit, the player is allowed to empty cups (reducing their height to $0$) rather than only removing a single unit of water. This is known as the \defn{cup flushing game} and has been studied for over $30$ years \cite{DietzSl87, DietzRa91, BenderFaKu19, LitmanMo09}. 

The \defn{greedy} algorithm, which always selects the fullest cup, has been a central research topic in every variation of the cup game.
\textbf{Greedy} is known to achieve optimal backlog $\Theta(\log n)$ in the cup flushing game due to a seminal $1987$ paper \cite{DietzSl87} by Dietz and Sleator.
Interestingly, this tight $\Theta(\log n)$ bound on backlog is robust, regardless of how much resource augmentation the player is given.
Their proof \cite{DietzSl87}, in under half a page, uses an elegant family of invariants, and can be extended to the canonical cup game, as was  independently discovered by~\cite{AdlerBeFr03}. These works show that in both settings the \textbf{greedy} algorithm achieves the exact optimal backlog\footnote{If the amount of water in a cup is allowed to be negative, the bound simplifies to $H(n)$.} of $H(n-1) + 1$, where $H(n)$ denotes the $n$th harmonic number, i.e., $H(n) = \sum_{i=1}^n 1/i$. Remarkably, this past work demonstrates that there is no separation between the optimal worst-case backlog of the cup game and the cup flushing game.

We introduce a new model, the \defn{semi-oblivious cup game}, which captures uncertainty in the player's estimate of the height of each cup, thus strengthening the adversary's hand. In contrast to most well-studied variations of the cup game, the semi-oblivious cup game weakens the player. Each time step $t$, after the adversary has distributed its unit of water, for every cup $j$, rather than being informed of the exact height $g_j(t)$ of cup $j$, the player instead receives an estimate $g'_j(t)$ of the height of cup $j$. This estimate must satisfy 
$g'_j(t) \in [g_j(t), c \cdot g_j(t)]$, where $c$ is a constant parameterizing the model.
The \defn{greedy} algorithm simply selects the cup with the maximal $g'_j(t)$.  By analyzing the semi-oblivious model, we can gain insight into settings in which the player is given imperfect information regarding the cup heights (i.e., the processor has imperfect information regarding the amount of work in each buffer), and thus the player may end up selecting a cup that is off by at most a constant factor.

In the semi-oblivious cup game model, the \textbf{greedy} algorithm always removes water from cups whose height is within an $O(1)$ factor of the fullest cup. What is unclear is how this affects the game. How important is having perfect information for achieving good backlog bounds? Does \textbf{greedy} continue to achieve a backlog of $O(\log n)$, or do the guarantees of the algorithm fundamentally require access to high-fidelity measurements of how much water is in the cups?

\subsection{Our Results}

Our first result is a tight analysis of the \textbf{greedy} algorithm in the semi-oblivious cup game. In this setting, the \textbf{greedy} algorithm always selects a cup whose height is within a constant multiplicative error $c$ of the fullest cup, and removes up to 1 unit of water from it. We show that \textbf{greedy} performs much worse in this setting than in the classical version of the game. In particular, we prove matching upper and lower bounds of the form $\Theta(n^{\frac{c-1}{c}})$ on the algorithm's backlog.

A natural question is whether the bound achieved by semi-oblivious \textbf{greedy} might improve with the help of resource augmentation, i.e., in the \defn{semi-oblivious cup flushing game}, where the player removes all of the water from some cup on each step.
Our second result shows that resource augmentation does, in fact, make a difference, allowing \textbf{greedy} to achieve a sub-polynomial bound on backlog. Specifically, we prove tight upper and lower bounds of the form $2^{\Theta(\sqrt{\log n})}$ on the backlog that \textbf{greedy} achieves in the semi-oblivious cup flushing game. 

These results demonstrate a separation between the performance of the \textbf{greedy} algorithm in the semi-oblivious cup game with and without resource augmentation. This contrasts with results for the canonical cup game, where resource augmentation has no effect on the optimal backlog~\cite{DietzSl87, AdlerBeFr03}.

These results may give the impression that the $O(\log n)$ bound achieved by \textbf{greedy} in the canonical cup game is actually quite fragile. Does \textbf{greedy} actually require perfect information to achieve the classic $O(\log n)$ backlog guarantee? To this end, we consider an additive-error version of the semi-oblivious cup game, in which the adversary's estimates are constrained by $f'_j(t) \in [f_j(t), c + f_j(t)]$. Past work \cite{BenderFaKu19} established that \textbf{greedy} achieves $O(\log n)$ backlog in the cup \emph{flushing} game, even with $O(1)$ additive error.  For the standard (non-flushing) version of the game, we find matching upper and lower bounds on \textbf{greedy}'s performance in the additive error setting of 
$$(c+1) \ln n \pm O(1) = \Theta(\log n).$$ 
This result is especially exciting since the known techniques for the cup flushing game, i.e., the potential-function style arguments of \cite{BenderFaKu19}, do not yield anything non-trivial in this setting. Our backlog bound can be viewed as a robustness result, demonstrating that \textbf{greedy} remains resilient against additive errors in the canonical cup game, but that its resiliency decays directly proportionally to the quantity $1 + c$.

All of our lower bounds in the semi-oblivious cup game extend to apply not just against the \textbf{greedy} algorithm, but against any, even randomized, algorithm under an adaptive adversary. Thus, as an immediate (and non-trivial) corollary, we get that \textbf{greedy} is asymptotically optimal in all of the above settings. \john{Probably could extend to an oblivious adversary against any randomized algorithm. Or probably not actually?}

\subsection{Other Related Work}

The \textbf{greedy} algorithm has also received considerable attention in the fixed-rate cup game, a variant of the cup game where the adversary must select some fixed rates for the cups ahead of time. Although the optimal backlog is $2 = O(1)$~\cite{BaruahCoPl96,GasieniecKlLe17,BaruahGe95,LitmanMo11,LitmanMo05,MoirRa99,BarNi02,GuanYi12,Liu69, LiuLa73}, the \textbf{greedy} algorithm is unable to improve in the fixed-rate setting over the cup game, still achieving backlog $\Theta(\log n)$~\cite{AdlerBeFr03}.

When both the fixed-rate restriction on the adversary is combined with giving the player resource augmentation, the result is the fixed-rate cup flushing game---this is also known in the literature as the \defn{bamboo garden trimming problem}, introduced by \cite{GasieniecKlLe17}.\footnote{The name stems from a rather creative metaphor in which $n$ bamboos grow, each at a fixed rate, and a gardener (who happens to be a panda) selects $1$ bamboo each day to chop down.}
The \textbf{greedy} algorithm has also received considerable attention in the bamboo setting, where the optimal backlog for any algorithm is known to be exactly  $2$~\cite{BaruahCoPl96,GasieniecKlLe17,BaruahGe95,LitmanMo11,LitmanMo05,MoirRa99,BarNi02,GuanYi12,Liu69, LiuLa73}. 
The best known upper bound in this setting was improved from the baseline of $O(\log n)$~\cite{DietzSl87, AdlerBeFr03} to $9$~\cite{BiloGuLe20}, and then to the current best bound of $4$ ~\cite{Kuszmaul22}. Nonetheless, greedy was recently proved to be non-optimal, with a lower bound on its backlog of $2.076$~\cite{JasinskaKuLe26}. 

Although \textbf{greedy} is not optimal in the bamboo setting (or even asymptotically optimal in the fixed-rate cup game), there is a simple exactly optimal algorithm for these settings known as the \textbf{Deadline-Driven Strategy}. First studied in the late 1960s by Liu and Layland \cite{Liu69, LiuLa73}, the \textbf{Deadline-Driven Strategy} selects the cup that will soonest achieve height $2$ of those that have height at least $1$. One of the consequences of the proofs of \cite{LitmanMo09} is that the algorithm is worst-case optimal. The \textbf{Deadline-Driven Strategy} was introduced and analyzed directly in the bamboo and fixed-rate cup game settings by \cite{Kuszmaul22}.

One additional line of research has achieved progressively better competitive ratios for the bamboo trimming problem, which is beneficial for instances of bamboo which may admit a solution with backlog less than $2$ \cite{GasieniecKlLe17, pinwheel2, pinwheel3, GasieniecJuKl24, HohneSt23, Kawamura24, mishra26, kleinbergmi26}. 

We note that while our lower bounds apply for adaptive adversaries against both deterministic and randomized algorithms, randomized algorithms for cup games against oblivious adversaries have also been studied in the literature \cite{BenderKu20, BenderFaKu19, Kuszmaul20, DietzRa91}, achieving backlog $O(\log\log n)$ even against `slightly' non-adaptive adversaries \cite{BenderKu20}. In the semi-oblivious setting, the concept of an oblivious adversary is somewhat murky in that an oblivious adversary is not inherently capable of providing adversarial estimates of cup heights that fall within the required multiplicative or additive error.

\section{Preliminaries}

We prove an inequality about the harmonic series that will be useful in our cup game analysis. See \cite{Pugh04} for related facts.
\begin{fact}
For fixed constant $c > -1$, $\prod _{i=1}^n(1+\frac{c}{i})=\Theta(n^c)$.
    \label{fact:harmonic}
\end{fact}
\begin{proof}
We will prove this fact by expressing the formula using the Gamma function and approximating it. First, recall the basic property of the Gamma function:

\[
\Gamma(z+1) = z \Gamma(z) \quad \text{for all positive real numbers } z.
\]

Additionally, we know that
\[
\Gamma(n) = (n-1)! \quad \text{for all positive integers } n.
\]

Using these facts, we can express the formula as
\begin{align*}
\prod_{i=1}^n \left( 1 + \frac{c}{i} \right) &= \frac{\prod_{i=1}^n (i+c)}{n!} \\
&= \frac{1}{n!} \prod_{i=1}^n \frac{\Gamma(i+c+1)}{\Gamma(i+c)} \\
&= \frac{\Gamma(n+c+1)}{\Gamma(c+1) \Gamma(n+1)} \\
&= \Theta\left( \frac{\Gamma(n+c+1)}{\Gamma(n+1)} \right).
\end{align*}
The last equation holds because $\Gamma(c+1)>0$ for $c>-1$. From the Lanczos approximation (see \cite{Pugh04}), we can approximate the Gamma function as
\[
\Gamma(z+1) = \Theta\left( z^{z+0.5} e^{-z} \right) \quad \text{as } |z| \to \infty.
\]

Therefore,

\begin{align*}
\Theta \left( \frac{\Gamma(n+c+1)}{\Gamma(n+1)} \right) &= \Theta\left( \frac{(n+c)^{n+c+0.5} e^{-n-c}}{n^{n+0.5} e^{-n}} \right) \\
&= \Theta\left( \left( 1 + \frac{c}{n} \right)^{n+0.5} (n+c)^{c} e^{-c} \right) \\
&= \Theta\left( e^c n^c e^{-c} \right) \\
&= \Theta(n^c).
\end{align*}

Finally, we note that in our model, the only information which the player receives are the estimates of the heights of the cups provided by the adversary. We note that the adaptive adversary always knows the exact height of each cup (since it is informed of the player's selected cup every time step). Finally, the greedy algorithm is only specified to select an arbitrary cup with maximum estimate, thus our upper bounds apply even with worst-case (adversarial) tie breaking. On the other hand, our lower bounds apply to all algorithms, including, of course, any variant of greedy with any tie breaking mechanism.

\end{proof}

\section{The Greedy Algorithm in the Semi-Oblivious Cup Game}
\seclabel{cup}

Recall that in the semi-oblivious cup game, the \textbf{greedy} algorithm can select any arbitrary cup within a factor of $c$ of the fullest cup. (As the adversary can select any of those cups to have the highest estimate.) In \secref{semi_cup_game}, we prove matching upper and lower bounds on \textbf{greedy}'s performance in the semi-oblivious cup game. In \secref{semi_flushing_game}, we provide matching upper and lower bounds on \textbf{greedy}'s performance in the semi-oblivious cup flushing game. Finally, in \secref{additive error}, we show that the performance of \textbf{greedy} in the additive error setting does not differ significantly from the original cup game.

\subsection{The Semi-Oblivious Cup Game} \label{sec:semi_cup_game}

We begin by introducing some notation that will be used throughout \secref{cup}. There are $n$ cups, and the game starts at time $0$ with all cups empty, and the adversary begins injecting water. We assume that for each time $t$, the adversary moves first, followed by the player. We let \(f_i(t)\) denote the height of cup \(i\) immediately before
the adversary's move in time step \(t\), and \(g_i(t)\) its height
immediately after the adversary's move. Thus, if the player selects
cup \(p_t\), then in the non-flushing game
\[
f_{p_t}(t+1)=\max\{g_{p_t}(t)-1,0\},
\qquad
f_j(t+1)=g_j(t)\quad\text{for }j\neq p_t.
\]
In the flushing game, \(f_{p_t}(t+1)=0\).
We denote the estimate provided to the player by the adversary by $g_i'(t) \in [g_i(t),c\cdot  g_i(t)]$ for constant $c > 1$.  Let $ c_i(t)$ and $ d_i(t)$ denote the index of the $i$th fullest cup at time $t$, and at the intermediate step of time $t$, respectively. We denote by $f(t, m) = (1/m) \sum_{i=1}^m f_{c_i(t)}(t)$ the average of the $m$ fullest cups at time $t$. We denote by $g(t, m) = (1/m) \sum_{i=1}^m g_{d_i(t)}(t)$ the average of the $m$ fullest cups at the intermediate step at time $t$. 
    
In the following theorem, we upper bound the backlog achieved by the \textbf{greedy} algorithm in the semi-oblivious setting.

\begin{theorem}
The \textbf{greedy} algorithm achieves a backlog of $O\left(n^{\frac{c-1}{c}}\right)$ in the semi-oblivious cup game for any constant $c > 1$.
\label{thm:multiplicative_upper}
\end{theorem}

We first show the below lemma. Our argument inductively utilizes a series of invariants, inspired by the style of \cite{DietzSl87}.

It suffices to consider only the times when the maximum cup height exceeds $c$. To bound the backlog at an arbitrary time $t_1$, we let $t_0\leq t_1$
denote the largest time such that $f(t_0,1)<c$, i.e., the last time prior to $t_1$ when all cups have height less than $c$. It is sufficient to show that in the interval $[t_0, t_1]$, the heights of all cups remain bounded by $O\bigl(n^{(c-1)/c}\bigr)$; below, we exclusively consider the interval $[t_0, t_1]$.
We may assume that $g_{d_1(t)}(t) \ge c$ for all $t \ge t_0$ in our interval $[t_0, t_1]$. From $t_0$ onward, the total amount of water in the system remains constant until $t_1$, since the player only selects cups whose heights are at least $1$.

\begin{lemma}
For all $t\geq t_0+1, 1\leq m \leq n-1$, we have either  
$$
g(t, m) \leq g(t-1, m)
$$
or  
$$
g(t, m)+\frac{c}{c-1} \leq \left(g(t-1, m+1)+\frac{c}{c-1}\right)\left(1 + \frac{c-1}{cm}\right).
$$
\label{lem:multiplicative_lem}
\end{lemma}

\begin{proof}
We denote the index of the cup selected by the player at time $t-1$ by $i$.

\paragraph{Case 1: \boldmath Cup $i$ is among the $m$ fullest cups at time $t$.} 
In this case, cup $i$ is one of the $m$ fullest cups at time $t$ even though the player just removed $1$ unit of water from it. This means that it certainly must have been one of the $m$ fullest cups before the player's turn. Thus we have
\begin{align*}
    f(t, m) = g(t-1, m) - \frac{1}{m},
\end{align*}
This follows because the player's action of reducing the height of cup $i$ by $1$ results in the sum of the $m$ fullest cups decreasing by $1$, since cup $i$ remains among the $m$ fullest cups. 
On the other hand, we know that 
\begin{align*}
    g(t, m) \leq f(t, m) + 1/m.
\end{align*}
Combining this with the previous equality, we obtain the desired inequality:
\[g(t, m) \leq g(t-1, m).\]

\paragraph{Case 2: \boldmath Cup \( i \) is not among the \( m \) fullest cups at time \( t \).}  

Since cup \( i \) was selected by the player, we have that 
\[
g_i(t-1) \geq \frac{g_{d_1(t-1)}(t-1)}{c},
\]
i.e., it was within a multiplicative factor \( c \) of the fullest cup. Then,
\begin{align*}
f(t, m) &= \frac{\sum_{i=1}^m f_{c_i(t)}(t)}{m} \\
&= \frac{\left( \sum_{i=1}^m f_{c_i(t)}(t) \right) + g_i(t-1) - g_i(t-1)}{m} \\
&\leq \frac{g(t-1, m+1)(m+1) - g_i(t-1)}{m}\\
&\leq \frac{g(t-1, m+1)(m+1) - \frac{g(t-1, m+1)}{c}}{m} \\
&= g(t-1, m+1)\left(1 + \frac{c-1}{cm}\right).
\end{align*}

The third line follows since we have 
\[
g(t-1, m+1)(m+1) \geq f(t, m) m + g_i(t-1).
\]
This holds because the $m$ fullest cups at time $t$ were all present with their exact heights at the intermediate step during $t-1$ (none of them were selected by the player), and so is cup $i$ with height $g_i(t-1)$ (which we know to be distinct from the $m$ fullest cups at time \( t \), by the case assumption). Therefore, the sum of the \( m+1 \) fullest cups at the intermediate step at time \( t-1 \) is at least \( m\cdot f(t, m) + g_i(t-1) \).

The fourth line follows from the inequality, 
\begin{align*}
g_i(t-1) &\geq \frac{g_{d_1(t-1)}(t-1)}{c} \\
&\geq \frac{g(t-1,m+1)}{c}.
\end{align*}
From the above result, we can deduce the desired inequality.
\begin{align*}
g(t,m)+\frac{c}{c-1}&\leq f(t,m)+\frac{1}{m}+\frac{c}{c-1}\\
&\leq \left(g(t-1, m+1)+\frac{c}{c-1}\right)\left(1 + \frac{c-1}{cm}\right)
\end{align*}
\end{proof}

Having completed the proof of the lemma, we now use the recursive guarantee of \lemref{multiplicative_lem} (as well as Fact \ref{fact:harmonic}) to complete the proof of the theorem.
\begin{proof}[Proof of \thmref{multiplicative_upper}]
We will prove the theorem by showing the inductive hypothesis,
\[
g(t,m) + \frac{c}{c-1} \leq \left(c+1 + \frac{c}{c-1}\right) \prod_{i=m}^{n-1} \left(1 + \frac{c-1}{ci}\right),
\]
for all $t \geq t_0$, $1 \leq m \leq n$. (We follow the standard convention that a product over an empty index set equals 1; that is, for $m=n$, we observe $\prod_{i=m}^{n-1} \left(1 + \frac{c-1}{ci}\right)=1$.) We label the parameterized inequality as $H(t, m)$.

For the base cases $H(t_0,m)$ and $H(t,n)$, we already know that
\[
g(t_0, m) + \frac{c}{c-1} \leq f(t_0, m) + \frac{1}{m} + \frac{c}{c-1} \leq c + 1 + \frac{c}{c-1} \quad \text{for all } 1 \leq m \leq n,
\]
and
\[
g(t, n) + \frac{c}{c-1} = g(t_0, n) + \frac{c}{c-1} \leq c + 1 + \frac{c}{c-1} \quad \text{for all } t \geq t_0,
\]
since there is conservation of total water in the $n$ cups on the interval $[t_0, t]$.

We will induct on $t - m$ to prove $H(t,m)$. When proving $H(t,m)$, we assume that $t > t_0$, $m < n$, and that $H(t-1,m)$ and $H(t-1,m+1)$ are true. By Lemma~\ref{lem:multiplicative_lem}, we know that $g(t,m) \leq g(t-1,m)$ or
\[
g(t,m) + \frac{c}{c-1} \leq \left(g(t-1, m+1) + \frac{c}{c-1}\right) \left(1 + \frac{c-1}{cm}\right).
\]
For the first case, the inequality
\[
g(t,m) + \frac{c}{c-1} \leq g(t-1,m) + \frac{c}{c-1} \leq \left(c+1 + \frac{c}{c-1}\right) \prod_{i=m}^{n-1} \left(1 + \frac{c-1}{ci}\right)
\]
holds by the inductive hypothesis $H(t-1,m)$. 

For the second case, the inequality
\begin{align*}
g(t,m) + \frac{c}{c-1} &\leq \left(g(t-1,m+1) + \frac{c}{c-1}\right) \left(1 + \frac{c-1}{cm}\right) \\
&\leq \left(c + 1 + \frac{c}{c-1}\right) \prod_{i=m}^{n-1} \left(1 + \frac{c-1}{ci}\right)
\end{align*}
holds by the inductive hypothesis $H(t-1,m+1)$.

Finally, we have proven the proposition $H(t,m)$ for all $t \geq t_0$, $1 \leq m \leq n$. By setting $m=1$, we conclude
\[
g(t,1) \leq \left(c+1 + \frac{c}{c-1}\right) \prod_{i=1}^{n-1} \left(1 + \frac{c-1}{ci}\right).
\]
Using Fact~\ref{fact:harmonic}, we know that 
\[
\prod_{i=1}^{n-1} \left(1 + \frac{c-1}{ci}\right) = O\left(n^{\frac{c-1}{c}}\right),
\]
so we conclude that the \textbf{greedy} algorithm achieves a backlog of $O\left(n^{\frac{c-1}{c}}\right)$.
\end{proof}

We now show a matching lower bound, proving that $\Theta(n^{\frac{c-1}{c}})$ is the asymptotically tight worst-case backlog for the \textbf{greedy} algorithm in the semi-oblivious cup game.

\begin{theorem}
The \textbf{greedy} algorithm achieves backlog $\Theta(n^{\frac{c-1}{c}})$ in the semi-oblivious cup game for any constant $c>1$. Furthermore, the lower bound of $\Omega(n^{\frac{c-1}{c}})$ applies to any (even randomized) algorithm with high probability in $\poly(n)$ steps against an adaptive adversary.
\end{theorem}

\begin{proof}
We describe a strategy for the adversary that forces backlog $\Theta(n^{(c-1)/c})$:
\begin{enumerate}
\item \textbf{Step 1.} While at least two cups are empty, the adversary pours $\frac{1}{2}$ unit of water into two empty cups. Since the player empties only one cup per step, at some intermediate time $t$ every cup has been filled to height $\frac{1}{2}$; hence at time $t+1$ there are $n-1$ cups of height $\frac{1}{2}$. 

We initially set $x = n-1$.

\item \textbf{Step 2.} At the start of stage $x$, let $A$ be the active set of $x$ cups, all of common actual height $h(x)$ immediately after the player's move, and write
\[
H_x:=\frac{cx}{cx-c+1}h(x).
\]
We maintain that the total height of the active cups immediately after every player move is at least $xh(x)$ and that no active cup ever exceeds $H_x$. Since
\[
(x-1)H_x+\frac{H_x}{c}=xh(x),
\]
every active cup has height at least $H_x/c$ whenever estimates are shown. Hence the adversary may assign the common estimate $H_x$ to every active cup.

The adversary chooses a distinguished cup $z\in A$ uniformly at random. In each time step, it puts as much as possible of its unit into cups of $A\setminus\{z\}$ that are below $H_x$, and then, if necessary, into $z$. If the active set cannot absorb the full unit, then every active cup reaches $H_x$, so after the player's move the stage succeeds. Otherwise the full unit is placed in $A$, and the post-player active total remains at least $xh(x)$.

A trial lasts $nx$ rounds and succeeds as soon as, immediately after a player move, at least $x-1$ active cups have height $H_x$. If this does not occur, the adversary chooses a fresh distinguished cup and begins another trial.

\item \textbf{Step 3.} On success, retain any $x-1$ active cups at height $H_x$, set
\[
h(x-1):=H_x,
\]
decrement $x$, and repeat Step 2 while $x\geq2$.
\end{enumerate}

We first argue that Step 2 terminates in polynomial time with high probability. For every cup $i \in A$, define the set of cups $S_i$ to be the non-active cups and $i$. A single trial has length $nx$ steps; during these steps, by the pigeonhole principle, some set of cups $S_i$ is emptied from at least $n$ times. Since the selected cup is uniform among the $x$ active cups, this fact holds for $S_z$ with probability at least $\frac{1}{x}$. We observe that during every time step in which the cup selected by the player is in $S_z$, the total volume of the $x-1$ other active cups increases by $1$. On the other hand, their total volume (and the total volume of all $x$ active cups) never decreases, since the adversary always adds its full unit of water to $A$. The total volume increase to the $x-1$ other active cups required to achieve height $H_x$ is at most $(x-1)(H_x - h(x)) = h(x) - \frac{H_x}{c} < h(x) < n$. Thus after $n$ such time steps in which cups from $S_z$ are selected, it is guaranteed that the adversary will have a successful trial. Otherwise, since the total volume of the active cups remains unchanged, the same volume argument will apply to future trials. 

If $i$ is the distinguished cup that the adversary selected, the target state is reached: every other active cup will reach height $H_x$. 
Hence within $n^2x^2$ steps a trial succeeds with probability $1-(1-\frac{1}{x})^{nx}=1-e^{-\Theta(n)}$, and a union bound over all values of $x$ shows that the whole process ends in $O(n^6)$ steps, except with exponentially small probability.

We now show that this strategy achieves backlog $\Theta(n^{\frac{c-1}{c}})$. Let $t(x)$ denote the time at which $x$ is decreased from $x+1$, and let $h(x)$ denote the common height of the $x$ retained active cups at time $t(x)$. By construction $h(x-1)=\frac{cx}{cx-c+1}h(x)$, and $h(n-1)=\frac{1}{2}$. Iterating the recurrence gives
\[
h(1) = \frac{1}{2} \cdot \prod_{i=2}^{n-1} \left(1+\frac{c-1}{ci-c+1}\right),
\]
and by Fact~\ref{fact:harmonic},
\begin{align*}
\frac{1}{2} \cdot \prod_{i=2}^{n-1} \left(1+\frac{c-1}{ci-c+1}\right)
&> \frac{1}{2} \cdot \prod_{i=2}^{n-1} \left(1+\frac{c-1}{ci}\right)\\
&= \Theta\!\left(n^{\frac{c-1}{c}}\right).
\end{align*}
\end{proof}

Thus, the \textbf{greedy} algorithm achieves worst-case backlog in the semi-oblivious cup game of $\Theta(n^{\frac{c-1}{c}})$, an exponential blow-up from the $\Theta(\log n)$ performance of the \textbf{greedy} algorithm in the canonical cup game. In short, selecting cups that are off by a constant factor from the fullest sinks the backlog guarantees of \textbf{greedy}. One consolation, however, is that we have proved the asymptotic optimality of \textbf{greedy} in this setting.

If the \textbf{greedy} algorithm fails in the semi-oblivious cup game, a natural follow-up question is: what happens in the semi-oblivious cup flushing game? Recall that, here, the player, benefiting from resource augmentation, empties the selected cup entirely, reducing its height to $0$. The lower bound described above no longer applies, given that it relies on a notion of `conservation of water,' that doesn't hold for the cup flushing game. Intuition suggests that the \textbf{greedy} algorithm should exhibit good performance when given unlimited resource augmentation, even in the semi-oblivious model. Even if it selects a cup that is off by a constant factor of $c$, the fact that the amount of water removed from the selected cup is unconstrained (and proportional to the height of the fullest cup) suggests that there is still hope. Indeed, in the next subsection, we will provide matching upper and lower bounds for \textbf{greedy} in this setting that constitute a significant improvement over $\poly(n)$.

\subsection{The Semi-Oblivious Cup Flushing Game}
\label{sec:semi_flushing_game}

\begin{theorem}
    The \textbf{greedy} algorithm achieves worst-case backlog $e^{O(\sqrt{\log n})}$ in the semi-oblivious cup flushing game.
\end{theorem}

\begin{proof}
We note that the overall style of our argument is similar to that of~\cite{BenderFaKu19}, whose potential function argument can be used to analyze any algorithm that selects cups within $O(1)$ \emph{additively} of the fullest cup, in the cup flushing game. They showed that such an algorithm will still achieve backlog $O(\log n)$. Our potential function is substantially different (the potential function of \cite{BenderFaKu19} is an exponential potential function).

Our main insight is to use the potential function
\[\Phi(t) = \sum_{j= 1}^n \int_{1}^{\max(f_j(t), 1)} e^{\alpha \log^2(x)} dx,\]
where $\log$ denotes the natural logarithm, and the constant $\alpha$ will be set later in the analysis to $1/(4 \log^2(c))$.

First, how does the adversary's injection of water into the cups at time $t$ affect $\Phi$? Suppose that the fullest cup during the intermediate step at time $t$ (after the adversary has injected water, but before the player has selected its cup) has height $\ell$. Assume that $\ell \geq 2$. Let $\Phi'(t)$ denote the potential during the intermediate step at time $t$. Then we have 
\[\Phi'(t) \leq \Phi(t) + e^{\alpha\log^2 \ell}.\]
The adversary is able to increase $\Phi$ by at most $e^{\alpha\log^2 \ell}$ during its turn. During the player's turn, some cup with height at least $\ell/c$ is selected to be flushed. The resulting decrease in potential gives us, for $\ell > c$,
\[\Phi(t+1) \leq \Phi'(t) - \int_1^{\ell/c} e^{\alpha\log^2 x} dx.\]
Substituting in for $\Phi'(t)$, we have
\begin{equation}  \label{eq:potential_change}
\Phi(t + 1) \leq \Phi(t) + e^{\alpha\log^2 \ell} - \int_1^{\ell/c} e^{\alpha\log^2 x} dx.
\end{equation}
Observe that
\begin{align*}
    e^{\alpha\log^2 \ell} &= e^{\alpha \log^2 c} + \int_c^\ell \frac{d}{dx} e^{\alpha\log^2 x} dx \\
    &= e^{\alpha \log^2 c} + \int_c^\ell \frac{2\alpha \cdot e^{\alpha\log^2 x} \log(x)}{x} dx \\
    &= e^{\alpha \log^2 c} + \int_1^{\ell/c} \frac{2\alpha \cdot e^{\alpha \log^2(cu)} \log(cu)}{cu} cdu \\
    &= e^{\alpha \log^2 c} + \int_1^{\ell/c} \frac{2\alpha \cdot e^{\alpha \left(\log^2(u) + 2 \log(c) \log(u) + \log^2(c)\right)} \log(cu)}{u} du \\
    &= e^{\alpha \log^2 c} + \int_1^{\ell/c} \frac{2\alpha \cdot e^{\alpha \log^2(u) + 2\alpha \log c \log(u) + \alpha \log^2(c)} \log(cu)}{u} du \\
    &= e^{0.25} + \int_1^{\ell/c} \frac{0.5/\log^2(c) \cdot e^{\alpha \log^2(u) + 0.5 \log(u)/\log(c) + 0.25} \log(cu)}{u} du \\
    &\leq e^{0.25} + \frac{0.5 e^{0.25}}{\log^2(c)} \int_1^{\ell/c} \frac{e^{\alpha \log^2(u)} \log(cu)}{\sqrt{u}} du \\
    &< e^{0.25} + \frac{0.5 e^{0.25}}{\log(c)} \int_1^{\ell/c} e^{\alpha\log^2(u)} du.
    \end{align*}
On the third line, we substitute $u = x/c$. On the sixth, we use our choice $\alpha = 1/(4\log^2(c))$. We use $c \geq e^2$. Note that for smaller values of $c$, our bound for $c = e^2$ clearly applies as well. To justify the final line, we note that for $u \geq 1$, we have that $\log(cu)/\sqrt{u} \leq \log(c)$ when $c\geq e^2 $. Substituting into equation \eqref{eq:potential_change}, we have
\begin{align*}
\Phi(t+1) &\leq \Phi(t) + e^{0.25} - \left(1 - \frac{0.5 e^{0.25}}{\log(c)}\right) \int_1^{\ell/c} e^{\alpha\log^2 x} dx \\
&\leq \Phi(t) + 2 - 0.5 \int_1^{\ell/c} e^{\alpha\log^2 x} dx
\end{align*}
For $\ell \geq 5c$, this yields $\Phi(t+1) < \Phi(t)$.

We know that initially, we have $\Phi(0) = 0$. Furthermore, if the tallest cup has height at least $5c = \Theta(1)$ during the intermediate step before $t$, then $\Phi(t) < \Phi(t-1)$. And we know that if the tallest cup has height less than $5c$ during the intermediate step before $t$, then $\Phi(t) < \sum_{j=1}^n \int_1^{5c} e^{\alpha\log^2(x)} dx = O(n)$. Thus we can conclude inductively that $\Phi(t)$ is bounded by $O(n)$ over all time steps. 

We observe now that the bound $\Phi(t) = O(n)$ implies backlog at most $e^{O(\sqrt{\log n})}$. We note that every cup's contribution to $\Phi(t)$ is non-negative. Let \(\Phi(t)\leq Cn\), where \(C\) depends only on \(c\). If
\(f_j(t)\geq2\), then
\[
Cn
\geq
\Phi(t)
\geq
\int_{f_j(t)-1}^{f_j(t)}
e^{\alpha\log^2x}\,dx
\geq
e^{\alpha\log^2(f_j(t)-1)}.
\]
Therefore
\[
f_j(t)
\leq
1+\exp\!\left(\sqrt{\frac{\log(Cn)}{\alpha}}\right)
=
e^{O(\sqrt{\log n})}.
\]
Thus no cup can achieve height more than $e^{O(\sqrt{\log n})}$ against the \textbf{semi-oblivious greedy} algorithm.

\end{proof}

\begin{theorem} \label{thm:flushing_lower_bound}
    The adversary can achieve backlog $e^{\Omega(\sqrt{\log n})}$ against the \textbf{greedy} algorithm in the semi-oblivious cup flushing game. Furthermore this lower bound extends to any deterministic algorithm.
\end{theorem}

\begin{proof}
We begin by proving the lower bound against greedy. We can assume throughout this proof (for the sake of simplicity of exposition) that $c \leq 3/2$ in the proof of this lower bound, since a larger $c$ only makes the adversary stronger. 

While at least two cups are empty, the adversary adds \(1/2\) unit to
each of two empty cups. After each round, the number of nonempty cups
increases by at least one, and every nonempty cup has height \(1/2\).
Thus the adversary can begin with \(n-1\) cups of height \(1/2\).

The adversary will apply the following strategy, given initially $k$ cups at height $x$. The adversary will select $\lfloor k/(x + 1) \rfloor$ cups. The adversary will add water to these cups until they achieve height $cx$. Meanwhile, the greedy player will be directed to flush from the remaining $k - \lfloor k/(x + 1) \rfloor$ cups. (These cups have height $x$, which is within a factor of $c$ of the adversary's selected cups, which are approaching $cx$.) Note that the number of time steps required by the adversary is $\lceil (c-1)x \lfloor k / (x + 1) \rfloor\rceil < (c-1)xk/(x +1)+1$. On the other hand, it will take the player $k - \lfloor k/(x + 1) \rfloor \geq kx/(x+1) \geq (c-1)kx/(x + 1)+1$ time steps to flush their assigned cups. The final inequality follows both by $c \leq 3/2$ and by assuming $k \geq 4$ (taking $k$ to be a sufficiently large constant that $(1 - (c-1)) kx/(x+1) \geq 1$). Thus this strategy allows the adversary to turn $k$ cups at height $x$ into $\lfloor k/(x+1) \rfloor$ cups at height $cx$. We observe that the adversary can repeatedly apply this strategy until at most $k = 4$ cups remain (and the adversary could terminate the procedure at any constant $k$ cups remaining without changing the asymptotic lower bound).

    If the adversary applies this strategy recursively, starting with $n-1$ cups at height $1/2$, and ending with $k=4$ cups, then the number of cups sufficient to achieve height $c^i/2$ is
    \begin{align*}
    O\left( \prod_{j=1}^i (\lceil c^j\rceil + 1) \right) &\leq O\left(\prod_{j=1}^{i} (c+2)^j\right) \\
     &= O\left((c+2)^{(i+1)i/2}\right) \\
     &= (c+2)^{O(i^2)}. 
    \end{align*}
    Thus, given $n$ cups, the adversary can achieve height $e^{\Omega(\sqrt{\log n})}$ against the greedy algorithm.

We observe that the above lower bound extends to any deterministic algorithm in the semi-oblivious cup flushing game. The adversary is still able to use the strategy of obtaining $\lfloor k/(x+1) \rfloor$ cups at height $cx$, given $k$ cups initially at height $x$. The adversary simply shows estimates of each of the $k$ cups as having height $cx$ (until the cup is emptied by the player), and selects the $\lfloor k/(x+1) \rfloor$ cups which it computes will be the last to be emptied by the deterministic player given these estimates. As shown above, the adversary has enough time to push the selected cups to height $cx$, and the rest of the argument follows identically.

\end{proof}

We now show how to extend the above lower bound to apply even to randomized algorithms.
\begin{theorem}  \label{thm:flushing_lower_bound_randomized}
     The adaptive adversary can force backlog $e^{\Omega(\sqrt{\log n})}$ against any randomized algorithm (with high probability in polynomially many time steps).
\end{theorem}
\begin{proof}
For the adaptive adversary to force $e^{\Omega(\sqrt{\log n})}$ backlog against a randomized algorithm, a slightly more involved randomized procedure is required, although the spirit of the lower bound remains the same. To modify the recursive procedure, given $k$ active cups at height $x$, the adaptive adversary uniformly at random selects $\lfloor k/(x+1) \rfloor$ cups and adds water uniformly to these cups for the next $\lceil(c-1) x \lfloor k/(x+1) \rfloor \rceil$ time steps, which is only enough time for the player to flush at most $kx/(x+1)$ of the cups. During this period, the adversary provides estimates such that all of the $k$ active cups appear uniform to the player. In fact, the player maintains the invariant that all currently active cups have never before been flushed by the player, and they have, at every time step, had identical estimated heights. Thus they are indistinguishable to the player. Condition on the player's random coins and expose membership in the
adversary's sample by deferred decisions. Whenever the player flushes
a previously unflushed active cup, that cup is a uniformly random
remaining member or nonmember of the sample, since all unflushed
active cups have had identical observation histories. Thus the expected number of cups flushed by the player from the adversary's sample is at most $x/(x+1) \lfloor k/(x+1) \rfloor$. Let $Y$ denote the number of unflushed cups from the adversary's sample. Then $\E[Y]$ is at least
\begin{align*}
\left\lfloor \frac{k}{x+1} \right\rfloor / (x+1) &\geq  \frac{k}{(x+1)^2} - \frac{1}{x+1} \\
&= \Omega\left(\frac{k}{(x+1)^2}\right).
\end{align*}
Applying a Chernoff bound, the number of unflushed cups will be at least $\Omega(k/(x+1)^2)$ with high probability in $n$, so long as $k/(x+1)^2 = \Omega(\log n)$ (for constant $c'$). Thus we can union bound over all steps of the recursion for which $k/(x+1)^2 = \Omega(\log n)$ with high probability. (If at any point in the recursion, the high probability bound does not hold, then the adversary simply declares failure.)

Let $x_j = c^j/2$, and suppose a successful stage maps $k$ active cups at height $x_j$ to at least $\alpha k/(x_j+1)^2$ cups at height $x_{j+1}$. We require $K_i = \Omega(\log n)$ active cups surviving the final successful stage (for the above Chernoff bound to apply), and we define the number of required active cups by 
\[K_j = \left \lceil \frac{2}{\alpha} (x_j + 1)^2 K_{j+1} \right \rceil.\]
We have
\begin{align*}
    K_0 &\leq O\left(\log n \left(\frac{2}{\alpha}\right)^i\prod_{j=0}^{i-1} (x_j + 1)^2 \right) \\
    &\leq O\left(\log n \frac{2}{\alpha}^i \prod_{j=0}^{i-1} c^{2j}/2\right) \\
     &= O\left(\log n \frac{1}{\alpha^i} c^{i(i-1)}\right) \\
     &= \exp{O(i^2 + \log\log n)}.
    \end{align*}
    Choosing $i = \delta \sqrt{\log n}$, for sufficiently small constant $\delta$, will ensure that $K_0 \leq n-1$. Thus the adversary can thus achieve backlog $x_i = e^{\Omega(\sqrt{\log n})}$ with only $n$ cups.
    
    Each stage lasted at most $n$ steps, and there are only $O(\sqrt{\log n})$ trials; the total length of the experiment is $\poly(n)$. Thus, given $n$ cups, the adaptive adversary can achieve height $e^{\Omega(\sqrt{\log n})}$ against an arbitrary randomized algorithm, with high probability in $n$.
\end{proof}

\section{The Additive Error Regime} \label{sec:additive error}

We use the same notation as used previously in \secref{cup}. Recall that in the additive error setting, we have $g_i'(t) \in [g_i(t), g_i(t)+c]$, for $c > 0$.

\begin{theorem}
The \textbf{greedy} algorithm achieves backlog $(c+1)\ln n+O(c+1)$ in the semi-oblivious additive error setting.
\label{thm:additive_upper}
\end{theorem}

To bound the backlog at an arbitrary time $t_1$, we let $t_0 \leq t_1$ denote the largest time such that $f(t_0, 1) < c+1$, i.e, the last time prior to $t_1$ when all cups have height less than $c+1$. It is sufficient to show that in the interval $[t_0, t_1]$, the heights of all cups remain bounded by $(c+1)\ln n + O(c+1)$; below, we exclusively consider the interval $[t_0, t_1]$.
We restrict our attention to the interval $[t_0, t_1]$. For every time $t$ satisfying \(t_0\leq t\leq t_1\), the cup selected by greedy has actual height at
least
\[
g(t,1)-c\geq1.
\]
Thus the player removes exactly one unit, and total water remains
constant throughout the interval. In particular,
\(g(t,n)=g(t_0,n)\).

\begin{lemma}
    For all $t \geq t_0+1, 1\leq m\leq n-1$, we have either
    \[g(t, m) \leq g(t-1, m)\]
    or
    \[g(t, m) \leq g(t-1, m+1) + \frac{c+1}{m}.\]
\label{lem:additive_lem}
\end{lemma}

\begin{proof}
    We denote the index of the cup selected by the player at time $t-1$ by $i$.

    \paragraph{Case 1: \boldmath Cup $i$ is among the $m$ fullest cups at time $t$}

    This case follows identically to that of Case 1 in Lemma $\ref{lem:multiplicative_lem}$.

    \paragraph{Case 2: \boldmath Cup $i$ is not among the $m$ fullest cups at time $t$}

    Since cup $i$ was selected by the player, we have that $g_i(t-1) \geq g_{d_1(t-1)}(t-1) - c$, i.e., it was within $c$ additively of the fullest cup.

    Then
    \begin{align*}
        f(t, m) &= \frac{\sum_{i=1}^m f_{c_i(t)}(t)}{m} \\
        &= \frac{\left( \sum_{i=1}^m f_{c_i(t)}(t) \right) + g_i(t-1) - g_i(t-1)}{m} \\
        &\leq \frac{g(t-1, m+1)(m+1) - g_i(t-1)}{m} \\
        &\leq \frac{g(t-1, m+1)(m+1) - g(t-1, m+1) + c}{m}\\
        &= g(t-1, m+1) + \frac{c}{m}.
    \end{align*}
    
    The third line follows since the $m$ fullest cups at time $t$ were all present with their exact heights at the intermediate step at time $t-1$ (since none of them were selected by the player), and so is cup $i$ with height $g_i(t-1)$, which we know to be distinct from the $m$ fullest cups at time $t$ by the case assumption. Therefore, the sum of the heights of the $m+1$ fullest cups at the intermediate step at time $t-1$ is at least $f(t, m)m + g_i(t-1)$.

    The fourth line follows from the inequality 
    \begin{align*}
    g_i(t-1) &\geq g_{d_1(t-1)}(t-1) - c \\
    &\geq g(t-1,m+1)-c.
    \end{align*}
    We can conclude that 
    \begin{align*}
    g(t,m) &\leq f(t,m)+\frac{1}{m} \\ &
    \leq g(t-1,m+1)+\frac{c+1}{m}.
    \end{align*}
\end{proof}

\begin{proof}[Proof of Theorem \ref{thm:additive_upper}]
We will prove the theorem by showing the following inductive hypothesis:
\[
g(t,m)\leq c+2+\sum_{i=m}^{n-1}\frac{c+1}{i}.
\]
For the base case \(t=t_0\), for every \(1\leq m\leq n\),
\[
\begin{aligned}
g(t_0,m)
&\leq f(t_0,m)+\frac1m\\
&<c+1+\frac1m\\
&\leq c+2.
\end{aligned}
\]
Moreover, conservation of total water gives
\[
g(t,n)=g(t_0,n)<c+2
\qquad\text{for every }t\in[t_0,t_1].
\]

Now, we will perform induction on $t - m$ to prove the hypothesis. Fix $t > t_0$ and $m < n$; we aim to prove the inductive hypothesis for this pair using the inductive assumption.

From Lemma~\ref{lem:additive_lem}, we know that
\[
g(t, m) \leq g(t - 1, m) \quad \text{or} \quad g(t, m) \leq g(t - 1, m + 1) + \frac{c + 1}{m}.
\]

\paragraph{\boldmath Case $1$: $g(t, m) \leq g(t-1, m)$.} We apply the inductive hypothesis to $(t - 1, m)$:
\[
g(t, m) \leq g(t - 1, m) \leq c + 2 + \sum_{i = m}^{n - 1} \frac{c + 1}{i}.
\]

\paragraph{\boldmath Case $2$: $g(t, m) \leq g(t-1, m+1) + (c+1)/m$.} We apply the inductive hypothesis to $(t - 1, m + 1)$:
\[
\begin{aligned}
g(t, m) &\leq g(t - 1, m + 1) + \frac{c + 1}{m} \\
&\leq c + 2 + \sum_{i = m + 1}^{n - 1} \frac{c + 1}{i} + \frac{c + 1}{m} \\
&= c + 2 + \sum_{i = m}^{n - 1} \frac{c + 1}{i}.
\end{aligned}
\]

In both cases, the desired bound holds. Therefore, we have proven the inductive hypothesis for all $t \geq t_0$ and $1 \leq m \leq n$.
This implies that 
\begin{align*}
    g(t, 1) &\leq c + 2 + (c+1) H_{n-1} \\
    &\leq 2c + 3 + (c + 1)\ln n \\
    &=(c+1)\ln n+O(c+1).
\end{align*}
\end{proof}

We now show a matching lower bound, proving that our upper bound is asymptotically optimal: the exact optimal backlog lies within $O(c+1)$ of $(c+1)\ln n$. This matches the perfect-information cup game, where the exact optimal backlog lies within $O(1)$ of $\ln n$ (i.e., the case $c=0$).

\begin{theorem}
The adversary can force backlog $(c+1)\ln n - O(c+1)$ in the semi-oblivious additive-error setting against the greedy algorithm. Furthermore, this lower bound extends to any (even randomized) algorithm, with high probability in polynomially many time steps, against an adaptive adversary.
\end{theorem}

\begin{proof}
We describe a strategy for the adversary that forces backlog $(c+1)\ln n - O(c+1)$. Initially set $x=n$ and $h(n)=0$.
\begin{enumerate}
\item \textbf{Step 1.} Let $A$ be the active set of $x$ cups, all of common height $h(x)$ immediately after the player's move, and write
\[
H_x:=h(x)+\frac{c+1}{x},
\qquad
q:=\lceil c+1\rceil.
\]
We maintain that the post-player total height of the active cups is at least $xh(x)$ and that no active cup exceeds $H_x$.

The adversary chooses a distinguished cup $z\in A$ uniformly at random. In each time step, it distributes as much as possible of its unit among cups of $A\setminus\{z\}$ that are below $H_x$, and then, if necessary, into $z$. If the active set cannot absorb the full unit, every active cup reaches $H_x$, so after the player's move the stage succeeds. Otherwise the full unit is placed in $A$, and the post-player active total remains at least $xh(x)$.

At every intermediate state of an unfinished round, the active total is therefore at least $xh(x)+1$. Since every active cup has height at most $H_x$, every active cup has height at least
\[
xh(x)+1-(x-1)H_x
=
H_x-c.
\]
Thus assigning estimate $H_x$ to every active cup is valid. A trial lasts $qx$ rounds and succeeds as soon as, immediately after a player move, at least $x-1$ active cups have height $H_x$. On failure, the adversary chooses a fresh distinguished cup and begins another trial.

\item \textbf{Step 2.} On success, retain any $x-1$ active cups at height $H_x$, set
\[
h(x-1):=H_x,
\]
decrement $x$, and repeat Step 1 while $x\geq2$.
\end{enumerate}

We first argue that Step 1 terminates in polynomial time with high
probability. For every cup $i\in A$, define the set of cups $S_i$ to be
the non-active cups and $i$. A single trial has length $qx$ steps;
during these steps, by the pigeonhole principle, some set $S_i$ is
selected by the player at least $q$ times.

We observe that the total volume increase required for the $x-1$ cups
in $A\setminus\{z\}$ to reach height $H_x$ is at most
\[
(x-1)H_x-\bigl(xh(x)-H_x\bigr)
=
x\bigl(H_x-h(x)\bigr)
=
c+1.
\]
Until these cups are within one unit of their target volume, their total
volume never decreases, and it increases by one whenever the player
selects a cup in $S_z$. Once they are within one unit of their target,
they remain within one unit, and the next selection from $S_z$ causes
all of them to reach height $H_x$. Thus $q=\lceil c+1\rceil$ selections
from $S_z$ guarantee a successful trial.

Since the distinguished cup is uniform among the $x$ active cups and is
hidden from the player by the common estimates, a trial succeeds with
probability at least $\frac{1}{x}$. Hence within $qnx^2$ steps some
trial succeeds with probability at least
\[
1-\left(1-\frac{1}{x}\right)^{nx}
=
1-e^{-\Theta(n)}.
\]
A union bound over all values of $x$ shows that the whole process
finishes in
\[
O(qn^4)=O((c+1)n^4)
\]
steps, except with exponentially small probability.

We now show that this strategy achieves backlog $(c+1)\ln n - O(c+1)$. Let $h(x)$ be the common starting height for each stage $x$. For $2\leq x\leq n$,
\begin{align*}
h(x-1)= h(x)+\frac{c+1}{x}
\end{align*}
(using $h(n)=0$ when $x=n$).
Summing this recurrence,
\begin{align*}
h(1) &= h(n)+\sum_{i=2}^{n}\frac{c+1}{i}\\
&= (c+1)\ln n - O(c+1).
\end{align*}
Since $h(1)$ is the height of the final active cup, this proves the desired lower bound on backlog.
\end{proof}

\input{acknowledgments}

\bibliographystyle{plain}
\bibliography{./all}
\clearpage
\appendix
\end{document}

%% file: packages.tex
\usepackage{amsthm}
\usepackage[T1]{fontenc}
\usepackage{relate}
\usepackage{pgfplots}
\usepackage{pgfplotstable}
\pgfplotsset{compat=1.8}
\usepackage[aboveskip=8pt,belowskip=-4pt]{caption}
\usepackage{subcaption}
\usepackage{xspace} 
\usepackage[margin=1in]{geometry}
\usepackage{textcomp}
\usepackage{balance}  
\usepackage{xcolor}
\usepackage{amsmath}
\usepackage{algorithm}
\usepackage{algorithmicx}
\usepackage{cite}
\usepackage{float}
\usepackage{siunitx}
\usepackage{graphicx}
\usepackage{xspace}
\usepackage{thm-restate}
\usepackage{tcolorbox}
\usepackage{balance}  
\usepackage{booktabs} 

\usepackage{marginnote} 
\usepackage{url}
\usepackage{amssymb}
\usepackage{enumitem}
\graphicspath{{./fig/}}
\usepackage{mathtools}
\usepackage[noend]{algpseudocode}
\newtheorem{theorem}{Theorem}
\newtheorem{lemma}{Lemma}

\usepackage{comment}
\usepackage{color}
\usepackage{hyperref}
\hypersetup{hidelinks}
\usepackage{cleveref}
\usepackage{tikz}
\usetikzlibrary{fadings}

\usepackage{titling}
\thanksmarkseries{arabic}

%% file: defines.tex
\newcommand{\punt}[1]{}

\newcommand{\E}{\mathbb{E}}

\makeatletter
\makeatletter
\def\@copyrightspace{\relax}
\makeatother

\newtheorem{fact}{Fact}

\newcommand{\defn}[1]       {{\textit{\textbf{\boldmath #1}}}}
 
\newcommand{\poly}{\mbox{poly}}

\date{}

\newcommand{\namedcomment}[3]{{\sf \scriptsize \color{#2} #1: #3}}

\newif\ifcomments
\commentsfalse

\newcommand{\john}[1]{\ifcomments {\namedcomment{john}{blue}{#1}} \fi{}}

\renewcommand{\epsilon}{\varepsilon}

\newcommand{\secref}[1]         {Section~\ref{sec:#1}}
\newcommand{\seclabel}[1]    {\label{sec:#1}}

\newcommand{\thmref}[1]         {Theorem~\ref{thm:#1}}

\newcommand{\lemref}[1]         {Lemma~\ref{lem:#1}}



%% file: titleauth.tex
\title{The Semi-Oblivious Cup Game: an Imperfect Information Setting}
\author{Kalina Jasińska\thanks{University of Cambridge. Email: kalina.a.jasinska@gmail.com} \and John Kuszmaul\thanks{MIT CSAIL. Email: john.kuszmaul@gmail.com} \and Gyudong Lee\thanks{MIT. Email: gyudong7@mit.edu}}

%% file: acknowledgments.tex


\subsubsection*{Acknowledgements} AI (GPT 5.6 pro) was used for both light editing and verifying proofs. All models, algorithms, analyses, potential functions, and invariants were human generated.

This research was supported by NSF Grant CCF-2330048, BSF Grant 2024233, and a Simons
Investigator Award. J. Kuszmaul was supported by an MIT Akamai Presidential Fellowship. The authors would like to thank Michael A. Bender for his mentorship and advice during this project.
